\documentclass[12pt,a4paper]{article}

\usepackage{amsmath}          
\usepackage{graphicx}         
\usepackage{cite}             
\usepackage{geometry}         
\usepackage{hyperref}         
\usepackage{authblk}
\usepackage{bm}
\usepackage{amsthm}
\usepackage{amsfonts}
\usepackage{booktabs}
\usepackage[ruled]{algorithm2e}

\newtheorem{theorem}{Theorem}
\newtheorem{lemma}{Lemma}

\allowdisplaybreaks

\title{Approximating N-Player Nash Equilibrium through Gradient Descent\footnote{The International Joint Conference on Theoretical Computer Science (IJTCS), 2020, Online Version.}}

\author[1]{Dongge Wang}
\author[2]{Xiang Yan}
\author[1]{Zehao Dou}
\author[2]{Wenhan Huang}
\author[1]{Yaodong Yang}
\author[1]{Xiaotie Deng}

\affil[1]{Peking University}
\affil[2]{Shanghai Jiao Tong University}
\affil[ ]{\texttt{dgwang96@pku.edu.cn}, 
\texttt{xyansjtu@163.com}, \texttt{zehaodou@pku.edu.cn}, 
\texttt{rowdark@sjtu.edu.cn}, \texttt{yaodong.yang@pku.edu.cn}, \texttt{xiaotie@pku.edu.cn}}

\date{}

\begin{document}

\maketitle

\begin{abstract}
Decoding how rational agents should behave in shared systems remains a critical challenge within theoretical computer science, artificial intelligence and economics studies. Central to this challenge is the task of computing the solution concept of games, which is Nash equilibrium (NE).
Although computing NE in even two-player cases are known to be PPAD-hard, approximation solutions are of intensive interest in the machine learning domain.
In this paper, we present a gradient-based approach to obtain approximate NE in N-player general-sum games.
Specifically, we define a distance measure to an NE based on pure strategy best response, thereby computing an NE can be effectively transformed into finding the global minimum of this distance function through gradient descent.
We prove that the proposed procedure converges to NE with rate $O(1/T)$ ($T$ is the number of iterations) when the utility function is convex.
Experimental results suggest our method outperforms Tsaknakis-Spirakis algorithm, fictitious play and regret matching on various types of N-player normal-form games in GAMUT.
In addition, our method demonstrates robust performance with increasing number of players and number of actions.
\end{abstract}


\section{Introduction}


Nash equilibrium \cite{nash1950equilibrium} describes the fix point in which no player has the motivation to change its own strategy.
Since John Nash's first introduction in the 1950s, Nash equilibrium has become a pivotal mathematical tool in modeling economic decision making. 
The vast developments of social activities on the Internet, along with a booming number of strategic participants, have made  Nash equilibrium imperative to real-world applications such as online marketing~\cite{niyato2016economics}, advertising~\cite{maehara2015budget}, crowdsourcing~\cite{kang2017sequential}, and sharing economy~\cite{amar2015game}.
In the machine learning domain, there is also  an increasing interest in computing Nash equilibrium. 
A recent example  is the GAN model~\cite{goodfellow2014generative}, whose training objective is in fact to find the Nash equilibrium in a $2$-player zero-sum game. %




Despite its importance, computing or approximating Nash equilibrium efficiently is still an open challenge.
In the sense of computational complexity, it was proven that finding a Nash equilibrium for a game is PPAD-complete, even it only involves two players \cite{chen2009settling}. 
As a result, recent studies have focused on finding approximation solution to NE.
Among the many solutions, fictitious play~\cite{berger2007brown} and regret matching~\cite{hart2000simple} are two representative approaches, as they are simple to implement and are shown to maintain sublinear convergence to NE in two-player zero-sum games.
Although these methods can be naturally applied to games with more than two players or two-player general-sum games, the convergence guarantee does not hold anymore.
For example,  fictitious play triggers chaos or cyclic behaviors in terms of its empirical frequency when applied to two-player general-sum games \cite{krishna1998convergence}; similar patterns are also discovered on regret matching method~\cite{mazumdar2019policy}.

One possible explanation about the non-convergence of fictitious play and regret matching methods on general-sum games is that they treat each player individually \cite{Papadimitriou2018from}. Precisely, they collect the empirical frequency of each player and adopt a simple update oracle to make each rational player behave as a self-adaptive agent.
The idea behind such a design is referred as an individual view, in which each player maximizes its own payoff regardless of the whole system.
This is not aligned with the principle of Nash equilibrium, which describes the stable state after competitive interaction between players, unless the played game is a two-player zero-sum game, each player knows its opponent's payoff as the reciprocal to its own payoff.

In this paper, we propose a gradient-based solver to compute approximate NE in N-Player general-sum games through a global view. Unlike aforementioned methods, our method optimizes all players' strategies simultaneously.
Based on the property of best-response strategy, we present a distance measure, named \textbf{NashD}, to compute the distance between the players' joint strategy profile and a Nash equilibrium. The zero value of \textbf{NashD} is proven to be a sufficient and necessary condition for reaching a Nash equilibrium.
The proposed procedure optimizes the joint strategy to NE in an iterative manner by applying gradient descent.
We show the rate of convergence of our approach by first proving the convexity for the distance function and the \emph{Lipschitz} continuity for its gradient.
To verify our approach, we compare the gradient-based method to the state-of-the-art NE solvers on a group of commonly used games called GAMUT~\cite{nudelman2004run}. We further test the robustness of the proposed algorithm by increasing the size of games to a large scale.


\section{Related Works}




The most famous algorithm for a 2-player game is the Lemke–Howson algorithm~\cite{lemke1964equilibrium}, which finds exact Nash equilibrium through solving the corresponding linear complementary problem.
The polynomial-time state-of-the-art approximation approach provided by \cite{tsaknakis2007optimization} guarantees a theoretical bound of $0.3393$.
They define an equivalent optimization problem in both players' joint strategy spaces and calculate the descent direction through a linear program.
Following their approach, \cite{deligkas2017computing} extend the idea to any polymatrix game, which contains multiple players and the payoff for each player is derived from multiple bimatrix games.
Their approach is proved to guarantee at most $1/2$-approximate Nash equilibrium in any polymatrix game.

Instead of theoretically proving the upper bound for the approximate ratio, many empirical studies present performance evaluation for popular Nash solvers.
Plenty of experiments are conducted on randomly-generated bimatrix~\cite{tsaknakis2008performance}, GAMUT bimatrix games~\cite{fearnley2015empirical}, and polymatrix games based on various graphs~\cite{deligkas2016empirical}.
Empirically, the gradient-based approaches mentioned above perform much better than the theoretical worst bound, especially better than those approaches based on other techniques, e.g. mixed integer programming~\cite{sandholm2005mixed} and heuristic search~\cite{porter2008simple}.
This motivates further studies for applying gradient-based algorithms to find Nash equilibrium in general games.



Other attempts for Nash solvers focus on games beyond normal-form or perfect information. 
For example, in extensive-form games, where players choose their actions alternatively, it is possible that participants only have imperfect information about opponents' joint strategy.
A series of works have been done to deal with such a difficulty by dividing the whole game into several subgames through a concept of information set (or further utilizing neural networks to do classification), each of which represents a circumstance condition on the imperfect information.
Then approaches, either based on those mentioned above~\cite{gilpin2007gradient,lockhart2019computing} or other methods including fictitious play and regret matching~\cite{zinkevich2008regret,heinrich2015fictitious}, can be applied to solve these subgames independently with some modification.
Since any extensive-form game can be uniquely transformed into an equivalent normal-form game, we only need to compare our approach with the basic methods of them on normal-form games, as we do in experiments.

\section{Discrete N-Player Normal-Form Game}\label{notation}

We denote a game instance as $G=(N, A, \bm{\bm{\sigma}}, u)$ and explain each item as follows.
\begin{itemize}
    \item $N$: the player set (or agent set). $N$ denotes the number of players as well for writing simplicity.
    \item $A$: the combinatorial action space. Each player has an action set $A_i$ with $i\in N$. For player $i$, We denote a specific action as $a_i\in A_i$ and the number of actions as $|A_i|$. An action is also referred to as a pure strategy which will be used interchangeably. A pure strategy profile $a=(a_1,a_2,\dots,a_N)$ represents one play of the game with player $i$ taking its corresponding action $a_i$. The action space $A=A_1\times A_2\times\dots\times A_N$ is a Cartesian product that contains all possible pure strategy profiles.
    \item $\bm{\bm{\sigma}}$: a mixed strategy profile. Each player can take stochastic action by adopting a mixed strategy, which is a probability distribution over its action set. We use $\bm{\sigma}_i\in \Delta(A_i)$ to denote a mixed strategy for player $i$ where $\Delta(A_i)$ represents the probability space. $\bm{\sigma}_i(a_i)$ denotes the probability assigned to action $a_i$ with $\sum_{a_i\in A_i}{\bm{\sigma}_i(a_i)}=1$. A mixed strategy profile $\bm{\sigma}=(\bm{\sigma}_1,\bm{\sigma}_2,\dots,\bm{\sigma}_N)$ represents one play of the game with player $i$ taking its corresponding strategy $\bm{\sigma}_i$. As pure strategies are special mixed strategy cases with 1.0 probability on the chosen action and zero on other actions, we often simply refer to a mixed strategy as a strategy. When analysing player $i$, other players' strategy is denoted by $\bm{\sigma}_{-i}=(\bm{\sigma}_1,\dots,\bm{\sigma}_{i-1},\bm{\sigma}_{i+1},\dots,\bm{\sigma}_N)$ which is a joint probability distribution over $A_{-i}=A_1\times\dots\times A_{i-1}\times A_{i+1}\times\dots\times A_N$.
    \item $u$: the utility function (also called payoff). In one play of the game, utility functions $u_1,u_2,\dots,u_N$ characterize the achieved value for each player. The utility of a mixed strategy $u_i(\bm{\sigma}_i, \bm{\sigma}_{-i})$ takes the expectation of the utility of pure strategies $u_i(a_i,a_{-i})$ as $u_i(\bm{\sigma}_i,\bm{\sigma}_{-i})=E_{a_i\sim\bm{\sigma}_i,a_{-i}\sim\bm{\sigma}_{-i}}{[u_i(a_i,a_{-i})]} =\sum_{a_i\in A_i}\sum_{a_{-i}\in A_{-i}}{\bm{\sigma}_i(a_i)\bm{\sigma}_{-i}(a_{-i})u_i(a_i,a_{-i})}$.
\end{itemize}

In this paper, we concentrate on N-Player general-sum games. For each player, the action space is discrete since the number of its pure strategies is countable. All game instances considered are in normal-form, in which all players take actions simultaneously for only one round.



\paragraph{Best response (BR).} Given a strategy profile $\bm{\sigma}$, the best response of player $i$ is the strategy (or strategies) which can achieve the maximal utility when other players' strategies are fixed. We denote the best response strategies as,
\begin{align}\label{BR}
    BR(\bm{\sigma}_{-i}) = \{ \mathop{\arg\max}\limits_{\bm{\sigma}_i\in \Delta(A_i)} u_i(\bm{\sigma}_i,\bm{\sigma}_{-i}) \}.
\end{align}

\paragraph{$\epsilon$-Approximate Nash equilibrium ($\epsilon$-NE).}
Nash equilibrium (NE) is the strategy profile when multiple perfectly rational agents interact in a game scenario.
A strategy profile $\bm{\sigma}^*$ becomes a Nash equilibrium if and only if each item $\bm{\sigma}_i^*$ is a best response strategy.
Due to the maximality of BR, no player can obtain better utility by changing its strategy unilaterally. 
Formally, $\bm{\sigma}^*$ is an NE if $\forall i\in N, \forall\bm{\sigma}_i\in \Delta(A_i), $
\begin{align}
        u_i(\bm{\sigma}^*_i,\bm{\sigma}^*_{-i}) - u_i(\bm{\sigma}_i,\bm{\sigma}^*_{-i}) \geq 0
\end{align}
Since in general computing an NE for any game is hard \cite{chen2009settling}, we focus on the corresponding approximate solution. Formally, $\bm{\sigma}^*$ is $\epsilon$-NE if $\forall i\in N, \forall\bm{\sigma}_i\in \Delta(A_i),$
\begin{align}\label{approNE}
        u_i(\bm{\sigma}_i,\bm{\sigma}^*_{-i}) - u_i(\bm{\sigma}^*_i,\bm{\sigma}^*_{-i})  \leq \epsilon
\end{align}
Generally speaking, the smaller $\epsilon$ value an algorithm achieves, the better approximation solution to NE it obtains.



\section{Distance Measure to Nash Equilibrium}

In this section, we analyze the property of Nash Equilibrium in zero-sum games, based on which we propose our measure \textbf{NashD} to describe the distance between any strategy profile and a Nash equilibrium.



For any player $i\in N$, consider a strategy profile $(\bm{\sigma}_i, \bm{\sigma}_{-i})$.
The utility of player $i$ can be written as

\begin{align}\label{MixExpectedUtility}
    u_i(\bm{\sigma}_i, \bm{\sigma}_{-i}) &= E_{a_i\sim \bm{\sigma}_i}[u_i(a_i,\bm{\sigma}_{-i})] \\
    &= \sum_{a_i\in A_i} \bm{\sigma}_i(a_i)u_i(a_i,\bm{\sigma}_{-i}).
\end{align}

For fixed opponents' strategy $\bm{\sigma}_{-i}$, $u_i(BR(\bm{\sigma}_{-i}), \bm{\sigma}_{-i})$ denotes the expected utility of player $i$ when taking any corresponding best response strategy.
Notice that even if there are multiple best response strategies, they lead to the same utility.
In fact, such a utility equals to the utility when player $i$ takes some pure strategy.

\begin{lemma}\label{BRPureMax}
    For each player $i$, given a fixed strategy $\bm{\sigma}_{-i}$,
    $$u_i(BR(\bm{\sigma}_{-i}), \bm{\sigma}_{-i}) = \max_{a_i\in A_i} u_i(a_i, \bm{\sigma}_{-i})$$
\end{lemma}
\begin{proof}
    
    By the definition of the best response strategy, 
    $$u_i(BR(\bm{\sigma}_{-i}), \bm{\sigma}_{-i}) \geq u_i(a_i, \bm{\sigma}_{-i}),\forall a_{i} \in A_{i}.$$
    Thus,
    $$u_i(BR(\bm{\sigma}_{-i}), \bm{\sigma}_{-i}) \geq \max_{a_i\in A_i} u_i(a_i, \bm{\sigma}_{-i}).$$
    Due to the fact $\sum_{a_i\in A_i} \bm{\sigma}_i(a_i)=1$, from Eqn. (\ref{MixExpectedUtility}) we have
    $$u_i(BR(\bm{\sigma}_{-i}), \bm{\sigma}_{-i}) \leq \max_{a_i\in A_i} u_i(a_i, \bm{\sigma}_{-i}).$$
    These together complete the proof.
\end{proof}


Based on this observation, we introduce \textbf{NashD} function:
$$\textbf{NashD}(\bm{\sigma})=\sum_{i\in N} \max_{a_i\in A_i} u_i(a_i, \bm{\sigma}_{-i}).$$
This is a non-negative function and its zero-points are precisely coincident with the Nash equilibrium.

\begin{theorem}\label{Ctheorem}
    In a zero-sum game $G$, for any strategy profile $\bm{\sigma}=(\bm{\sigma}_i,\bm{\sigma}_{-i})$, $\textbf{NashD}(\bm{\sigma}) \geq 0$, and $\bm{\sigma}$ is an NE if and only if $\textbf{NashD}(\bm{\sigma}) = 0$.
\end{theorem}
\begin{proof}
$G$ is zero-sum means $\sum_{i\in N} u_i(\bm{\sigma}_i, \bm{\sigma}_{-i}) = 0$.
Then by Lemma \ref{BRPureMax},
$$\sum_{i\in N} \max_{a_i\in A_i} u_i(a_i, \bm{\sigma}_{-i})\geq \sum_{i\in N} u_i(\bm{\sigma}_i, \bm{\sigma}_{-i})= 0. $$

If $\bm{\sigma}$ is not an NE, there must exist player $k$ and $\bm{\sigma}^{\prime}_k$ with

$$u_k(\bm{\sigma}^{\prime}_k,\bm{\sigma}_{-k}) > u_k(\bm{\sigma}_k, \bm{\sigma}_{-k}).$$
Thus,
$$u_k(BR(\bm{\sigma}_{-k}), \bm{\sigma}_{-k}) > u_k(\bm{\sigma}_k, \bm{\sigma}_{-k}).$$
And for any other player $r\in N\setminus \{k\}$,
$$u_r(BR(\bm{\sigma}_{-r}), \bm{\sigma}_{-r}) \geq u_r(\bm{\sigma}_r, \bm{\sigma}_{-r}).$$
Combining them we have
$$\sum_{i\in N} u_i(BR(\bm{\sigma}_{-i}),\bm{\sigma}_{-i}) > \sum_{i\in N} u_i(\bm{\sigma}_i, \bm{\sigma}_{-i}) = 0,$$
that is equivalently $\textbf{NashD}(\bm{\sigma}) > 0$.

If $\bm{\sigma}$ is an NE, it means $\forall i\in N$, $\bm{\sigma_i}\in BR(\bm{\sigma_i})$.
Then 
$$\sum_{i\in N}u_i(BR(\bm{\sigma}_{-i}),\bm{\sigma}_{-i}) =\sum_{i\in N} u_i(\bm{\sigma}_i, \bm{\sigma}_{-i}) =0,$$
that is equivalently $\textbf{NashD}(\bm{\sigma})=0$.
\end{proof}





\section{Distance Minimization over Strategy Space}
In this section, we propose a gradient-based method for finding approximate Nash equilibrium in players' joint strategy space of any given N-Player general-sum game. 
Furthermore, we present the convergence analysis of this gradient descent procedure.

\subsection{The Gradient Descent Algorithm}
Now we introduce the optimization procedure to find an approximate equilibrium based on the $\textbf{NashD}$ function.



The concrete process of our method is presented in Alg. \ref{alg:NashD_GD}. 
Firstly, in the input game $G$, the summation of all players' utilities may not always be zero, which is also referred to as a general-sum game. Thus, we need first convert it to a zero-sum game $G'$ by adding a fictitious player. The new player's payoff is always the reciprocal value of summing all other players' utility in $G$. 
After conversion, the algorithm randomly selects an initial point in the solution space by sampling randomized real numbers. 
On each iteration of the gradient descent, the origin real-number strategy vector will be projected as a probability profile by applying the softmax function. 
Then compute the \textbf{NashD} measure and its gradient. 
The gradient descent subprogram on \textbf{NashD} function executes for a fixed number of iterations and outputs the ending strategy profile, which provides an approximate Nash equilibrium.
\begin{algorithm}[t]
\caption{NashD gradient descent for N-player general-sum game}\label{alg:NashD_GD}
\LinesNumbered
\KwIn{Game $G=(N, \{A_i\}_{i=1}^N, \{u_i(\cdot)\}_{i=1}^{N})$, maximal iteration rounds $T$, learning rate $\alpha$}
\KwOut{An approximate Nash equilibrium strategy profile $\bm{\sigma}^*$}
Add a fictitious player who has a single action $a_{N+1}$ and obtain utility $u_{N+1}(\bm{\sigma},a_{N+1})=-\sum_i^N u_i(\bm{\sigma})$ \;
Randomly initialize strategy vector $\bm{\sigma}^{\prime,0}_i\in \mathbb{R}^{|A_i|}$ for $i=1,\dots,N+1$\  \;
\For{$t < T$}{
    Softmax on $\bm{\sigma}^{\prime,t}$ to get the strategy profile $\bm{\sigma}^t$, $\forall i$, $\bm{\sigma}^t_i(a_{ik})=\frac{Exp(\bm{\sigma}^{\prime,t}(a_{ik}))}{\sum_{k=1}^{|A_i|}Exp(\bm{\sigma}^{\prime,t}(a_{ik}))}$\;
    Compute \textbf{NashD} funtion $\textbf{NashD}(\bm{\sigma}^t)=\sum_i^{N+1}{\max_k{u_i(a_i^k,\bm{\sigma}_{-i}^t)}}$ \;
    Update $\bm{\sigma}^{\prime,t+1}(a_i^k) = \bm{\sigma}^{\prime,t}(a_i^k) - \alpha \cdot \nabla_{\bm{\sigma}^{\prime,t}(a_i^k)}\textbf{NashD}$ \;
}
Let strategy profile $\bm{\sigma}^* = \mbox{softmax}(\bm{\sigma}^{\prime,T})$
\end{algorithm}


\subsection{Convergence Analysis}
For convenience, we directly analyze the convergence of the strategy profile $\bm{\sigma}$. 
The proof consists of three parts:
we first show the convexity of each player's utility function implies the convexity of the $\textbf{NashD}$ function;
then we show its gradient is \emph{Lipschitz} continuous; finally, our approach will converge to a local optimum of the $\textbf{NashD}$ function, which follows directly from the convergence of the gradient descent algorithm~\cite{chong2013introduction}.
Note that we only show the convergence of our algorithm under a convex assumption, which is a common assumption in gradient-based literature. 
Although the convergence without any convex assumption is not theoretically proved, 
our experiments on GAMUT multi-player games show that our approach still converges without the convex condition for utility functions.
Our algorithm converges and obtains well-approximated equilibrium for both game classes with convex and general utility functions.

\begin{theorem}\label{thm:NashD_cvg}
In a game $G$ where the utility function $u_i$ of each player $i$ is convex, there exists a constant $L_{G}$ that, with the learning rate $\gamma_t\leq 1/L_{G}$, the gradient descent algorithm locally minimizes function $\textbf{NashD}(\bm{\sigma})$ with convergence rate $O(L_{G}/T)$ where $T$ is the number of iterations.
\end{theorem}

To prove Theorem~\ref{thm:NashD_cvg}, we first show the \textit{Lipshitz} continuity of the gradient of the $\textbf{NashD}(\bm{\sigma})$ function under the convexity assumption on each player's utility function.

\begin{lemma}\label{lemma:convex_l_continous}
In a game $G$ where the utility function $u_i$ of each player $i$ is convex, the function $\textbf{NashD}(\bm{\sigma})$ is convex and its gradient $\nabla_{\bm{\sigma}} \textbf{NashD}$ is \emph{Lipschitz} continuous with constant $L > 0$.
\end{lemma}

\begin{proof}
By definition, we discuss the function
$$\textbf{NashD}(\bm{\sigma}) = \sum\limits_{i\in N} \max\limits_{a_i\in A_i}u_i(a_i,\bm{\sigma}_{-i}):= \sum\limits_{i\in N}\textbf{NashD}_i(\bm{\sigma})$$
by analyzing the convexity of each item in the summation.
Note that for each $i$, the item 
\begin{align}\label{eqn:nashd_separate}
\textbf{NashD}_i(\bm{\sigma}) = &\max\limits_{a_i\in A_i}u_i(a_i,\bm{\sigma}_{-i}) \\
= &\max\limits_{a_i\in A_i}\sum\limits_{a_{-i}\in A_{-i}}u_i(a_i,a_{-i})\bm{\sigma}_{-i}(a_{-i})    
\end{align}
where $\sum_{a_{-i}\in A_{-i}}\bm{\sigma}_{-i}(a_{-i}) = 1$.
In other words, it is always a convex combination of utilities when other players plays pure strategies.

Now considering two joint policy, $\bm{\sigma}^{1} = (\bm{\sigma}^1_{1},\dots,\bm{\sigma}^1_{N})$ and $\bm{\sigma}^{2} = (\bm{\sigma}^2_{1},\dots,\bm{\sigma}^2_{N})$, and any $\lambda \in [0,1]$.
For each player $i$, each fixed action $a_i$ can be regarded as a special mixed strategy $\bm{\sigma}_1'$, then we have
\begin{align*}
&u_i(a_i,\lambda\bm{\sigma}^1_{-i}+(1-\lambda)\bm{\sigma}^2_{-i})\\ 
=&u_i(\lambda(\bm{\sigma}_1',\bm{\sigma}^1_{-i})+(1-\lambda)(\bm{\sigma}_1',\bm{\sigma}^2_{-i}))\\ \nonumber
\leq & \lambda u_i(\bm{\sigma}_1',\bm{\sigma}^1_{-i}) + +(1-\lambda)u_i(\bm{\sigma}_1',\bm{\sigma}^2_{-i})\\
=& \lambda u_i(a_i,\bm{\sigma}^1_{-i})+ (1 - \lambda) u_i(a_i,\bm{\sigma}^2_{-i})
\end{align*}
by the convexity of the utility function. 
Correspondingly,
\begin{align*}
    &\textbf{NashD}_i(\lambda \bm{\sigma}^1 + (1 - \lambda) \bm{\sigma}^2) \\
    \leq& \max\limits_{a_i\in A_i}\sum\limits_{a_{-i}\in A_{-i}}u_i(a_i,a_{-i})[\lambda \bm{\sigma}^1(a_{-i}) + (1 - \lambda) \bm{\sigma}^2(a_{-i})]\\
    =& \max\limits_{a_i\in A_i} [\lambda\cdot\sum\limits_{a_{-i}\in A_{-i}}u_i(a_i,a_{-i})\bm{\sigma}^1(a_{-i}) \\
    &~~~~~~~~~ + (1 - \lambda)\cdot \sum\limits_{a_{-i}\in A_{-i}}u_i(a_i,a_{-i})\bm{\sigma}^2(a_{-i})]\\
    \leq & \max\limits_{a_i\in A_i} [\lambda\cdot \max\limits_{a^1_i\in A_i}\sum\limits_{a_{-i}\in A_{-i}}u_i(a^1_i,a_{-i})\bm{\sigma}^1(a_{-i}) \\
    &~~~~~~~~~ + (1 - \lambda)\cdot \max\limits_{a^2_i\in A_i}\sum\limits_{a_{-i}\in A_{-i}}u_i(a^2_i,a_{-i})\bm{\sigma}^2(a_{-i})]\\
    =& \max\limits_{a_i\in A_i} [\lambda\cdot \textbf{NashD}_i(\bm{\sigma}^1) + (1 - \lambda)\cdot\textbf{NashD}_i(\bm{\sigma}^2)]\\
    =& \lambda\cdot \textbf{NashD}_i(\bm{\sigma}^1) + (1 - \lambda)\cdot\textbf{NashD}_i(\bm{\sigma}^2).
\end{align*}
This complete the proof for the convexity.

Now we prove the \emph{Lipschitz} continuity for the gradient $\nabla_{\bm{\sigma}} \textbf{NashD}$. We first denote
 $$U=\max_{i,a_i\in A_i,a_{-i}\in A_{-i}}|u_i(a_i,a_{-i})|$$
 as the maximum utility any player may obtain. 
(By Lemma~\ref{BRPureMax}, we only need to consider pure strategies.)
In fact, the joint policy $\bm{\sigma}$ is in a subspace of
$[0,1]^{|A_1|+ \cdots +|A_N|}$.
For readability, we specify a notation for the joint policy $\bm{\sigma} = (\bm{\sigma}_{11},\dots,\bm{\sigma}_{1|A_1|},\dots,\bm{\sigma}_{N1},\dots,\bm{\sigma}_{N|A_N|})$ with $\sum_{j}\bm{\sigma}_{ij}=1$ for each $i$.
$a_{il}\in A_i$ denotes the $l$-th pure strategy in $A_i$, which is represented by an one-hot vector denoted by $\vec{a_{il}}$.

Next, we prove that $||\nabla \textbf{NashD}(\bm{\sigma}^1) - \nabla \textbf{NashD}(\bm{\sigma}^2)|| \leq L||\bm{\sigma}^1 - \bm{\sigma}^2||$ for some constant $L < \infty$.
\footnote{Here we only consider $1$-norm as $||\cdot||$. One may define \emph{Lipschitz} continuity in $2$-norm, while they are equivalent except a constant.}
We do so through analyzing $\textbf{NashD}_i$ for each $i$ first.
According to Eqn. (\ref{eqn:nashd_separate}), the value of the $\bm{\sigma}_{jk}$-th dimension of the gradient $\nabla_{\bm{\sigma}} \textbf{NashD}_i$ is either $0$ (if $i=j$) or an utility corresponding to a group of strategies (if $i\neq j$).
We focus on the later case, and specify the notation $(\vec{a_{jk}},\bm{\sigma}_{-(i,j)})$ represent the joint strategy of players except player $i$, where player $j$ plays its $k$-th action while other players play the same strategies as $\bm{\sigma}_{-i}$.
In precise, 
\begin{align*}
\frac{\partial \textbf{NashD}_i}{\partial \bm{\sigma}_{jk}}|_{\bm{\sigma}} = u_i(\vec{a_{il^*}},\vec{a_{jk}},\bm{\sigma}_{-i,-j})
\end{align*}
where $l^*=\mbox{argmax}_{l\in \{1,2,\dots,|A_i|\}}u_i(\vec{a_{il}},\vec{a_{jk}},\bm{\sigma}_{-(i,j)})$.

Specifically for $\bm{\sigma}^1$ and $\bm{\sigma}^2$, we have
\begin{align*}
& |u_i(\vec{a_{il^*}},\vec{a_{jk}},\bm{\sigma}^1_{-i,-j}) -u_i(\vec{a_{il^*}},\vec{a_{jk}},\bm{\sigma}^2_{-i,-j})| \\
\leq & |u_i(\vec{a_{il^*}},\vec{a_{jk}}, \bm{\sigma}^1_{1}, \bm{\sigma}^2_{-(i,j,1)}) -u_i(\vec{a_{il^*}},\vec{a_{jk}}, \bm{\sigma}^2_{1}, \bm{\sigma}^2_{-(i,j,1)})|\\
+ & |u_i(\vec{a_{il^*}},\vec{a_{jk}}, \bm{\sigma}^1_{1}, \bm{\sigma}^1_{-(i,j,1)}) -u_i(\vec{a_{il^*}},\vec{a_{jk}}, \bm{\sigma}^1_{1}, \bm{\sigma}^2_{-(i,j,1)})|\\
\leq & 2U \cdot \max\limits_{l}|\bm{\sigma}^1_{1l}-\bm{\sigma}^2_{1l}| \\
+ & |u_i(\vec{a_{il^*}},\vec{a_{jk}}, \bm{\sigma}^1_{1}, \bm{\sigma}^1_{-(i,j,1)}) -u_i(\vec{a_{il^*}},\vec{a_{jk}}, \bm{\sigma}^1_{1}, \bm{\sigma}^2_{-(i,j,1)})|\\
\leq & \dots 
\leq 2U\cdot \sum\limits_{t\neq i,j}\max\limits_{l}|\bm{\sigma}^1_{tl}-\bm{\sigma}^2_{tl}|
\leq 2U \cdot ||\bm{\sigma}^1 - \bm{\sigma}^2||
\end{align*}
Then we have 
\begin{align*}
    &||\nabla \textbf{NashD}(\bm{\sigma}^1) - \nabla \textbf{NashD}(\bm{\sigma}^2)|| \\
    \leq~& 2UN\cdot (|A_1|+\ldots+|A_N|) \cdot ||\bm{\sigma}^1 - \bm{\sigma}^2||.
\end{align*}
This complete the proof for the \emph{Lipschitz} continuity.
\end{proof}

Finally, we can prove Theorem~\ref{thm:NashD_cvg}.

\begin{proof}
By Lemma~\ref{lemma:convex_l_continous}, the function $\textbf{NashD}(\bm{\sigma})$ is $L$-\emph{Lipschitz} continuous for a constant $L\leq2UN\cdot (|A_1|+\dots+|A_N|)$. According to standard analysis as~\cite{chong2013introduction}, with the learning rate $\gamma_{t} \leq 1 / L$, our approach will converge to a local optimum of the \textbf{NashD} function, and the convergence rate of the gradient descent algorithm is $O(L/T)$ where $T$ is the number of iterations.
\end{proof}


\section{Experiments}
To demonstrate the performance of our gradient-based approach, we test it in a series of games comparing with both the state-of-the-art and the widely applied baseline algorithms.

\subsection{Setup and Baseline Algorithms}


In all the experiments, we show the $\epsilon$-approximation value across our approach and baselines. As defined in Section \ref{notation}, we regard smaller approximation results as better solutions to NE. To make the numerical comparison consistent, we normalize the utility functions to $u_i: \bm{\sigma}  \rightarrow [0,1]$ for all tested game instances. 
The number of iteration rounds is set to be $T=1000$ for all compared algorithms. 
For our NashD gradient descent algorithm, we adopt a decay schedule on the learning rate which initializes $\alpha=0.5$ and multiplies the decay rate 0.8 every 100 iterations.

\paragraph{Tsaknakis-Spirakis (TS) algorithm.}
The algorithm proposed by \cite{tsaknakis2007optimization} (TS for short) is proven to be obtain state-of-the-art theoretical approximation bound of 0.3393 for 2-player general-sum games in polynomial time.
Note that the bound is for the worst cases and we will see in our experiments that the practical performance of TS is much better than this bound. 
However, TS cannot be applied to games with more than two players.
We set the TS hyper-parameter $\delta=0.001$ which follows the same setting in the origin empirical testing~\cite{tsaknakis2008performance}. 

\paragraph{Fictitious play (FP).} We implement ficititious play (FP for short) described by \cite{berger2007brown}. FP is known to converge for restricted game structures such as zero-sum games and potential games.
\cite{conitzer2009approximation} points out that though without convergence guarantee, FP outputs good approximation solutions in more general games.

\paragraph{Regret matching (RM).} We implement regret matching (RM for short) described by \cite{hart2000simple}. In 2-player zero-sum games, RM converges to a $2\epsilon$-approximate Nash equilibrium where $\epsilon$ is the time-averaged regret~\cite{cesa2006prediction}. 

\subsection{Effectiveness and Robustness}

\begin{figure*}[t]
\center
\begin{minipage}{0.48\textwidth}
\centering
\includegraphics[width=\columnwidth]{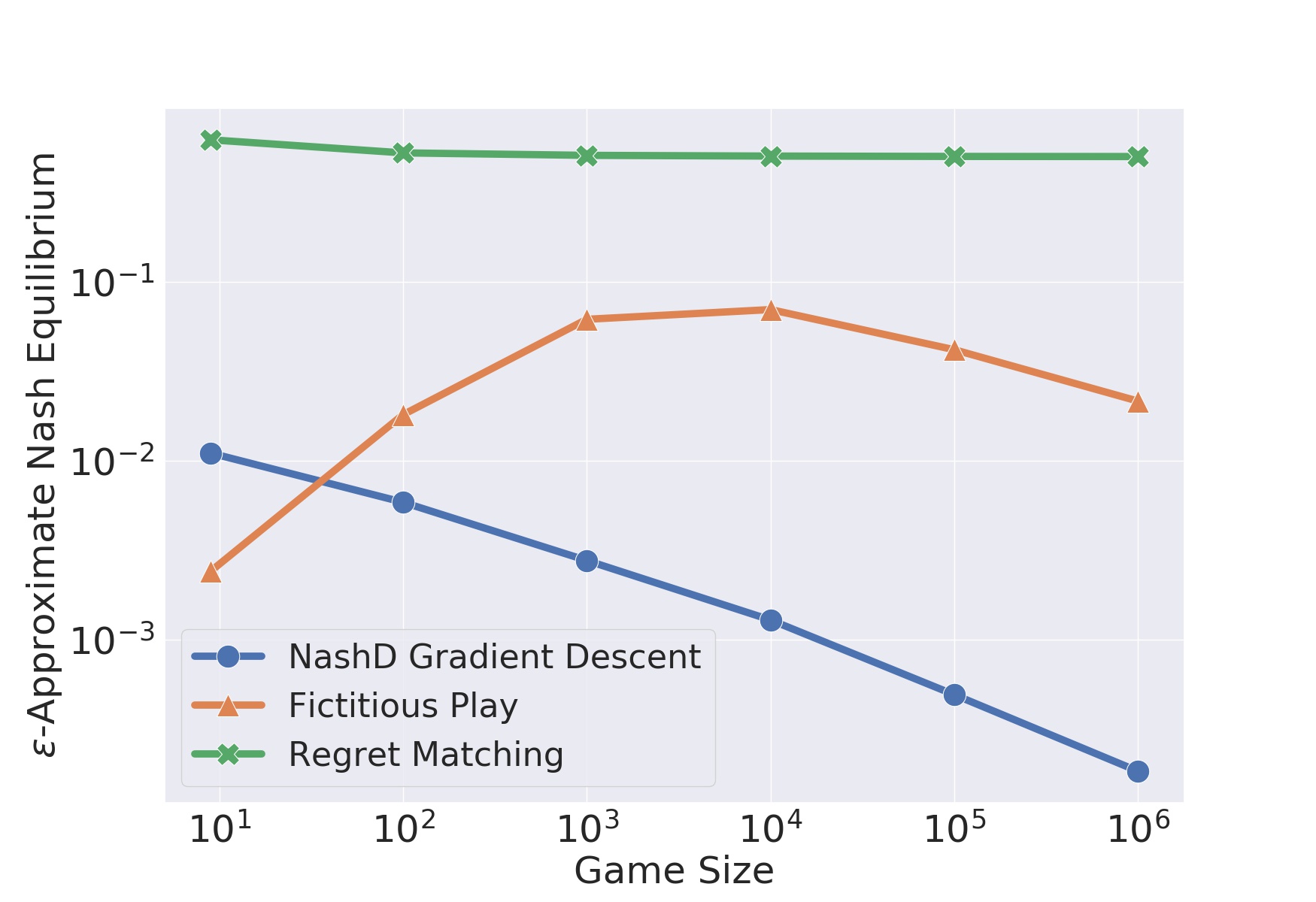}
\footnotesize{(a)}
\end{minipage}
\label{fig:gamesize}
\begin{minipage}{0.48\textwidth}
\centering
\includegraphics[width=\columnwidth]{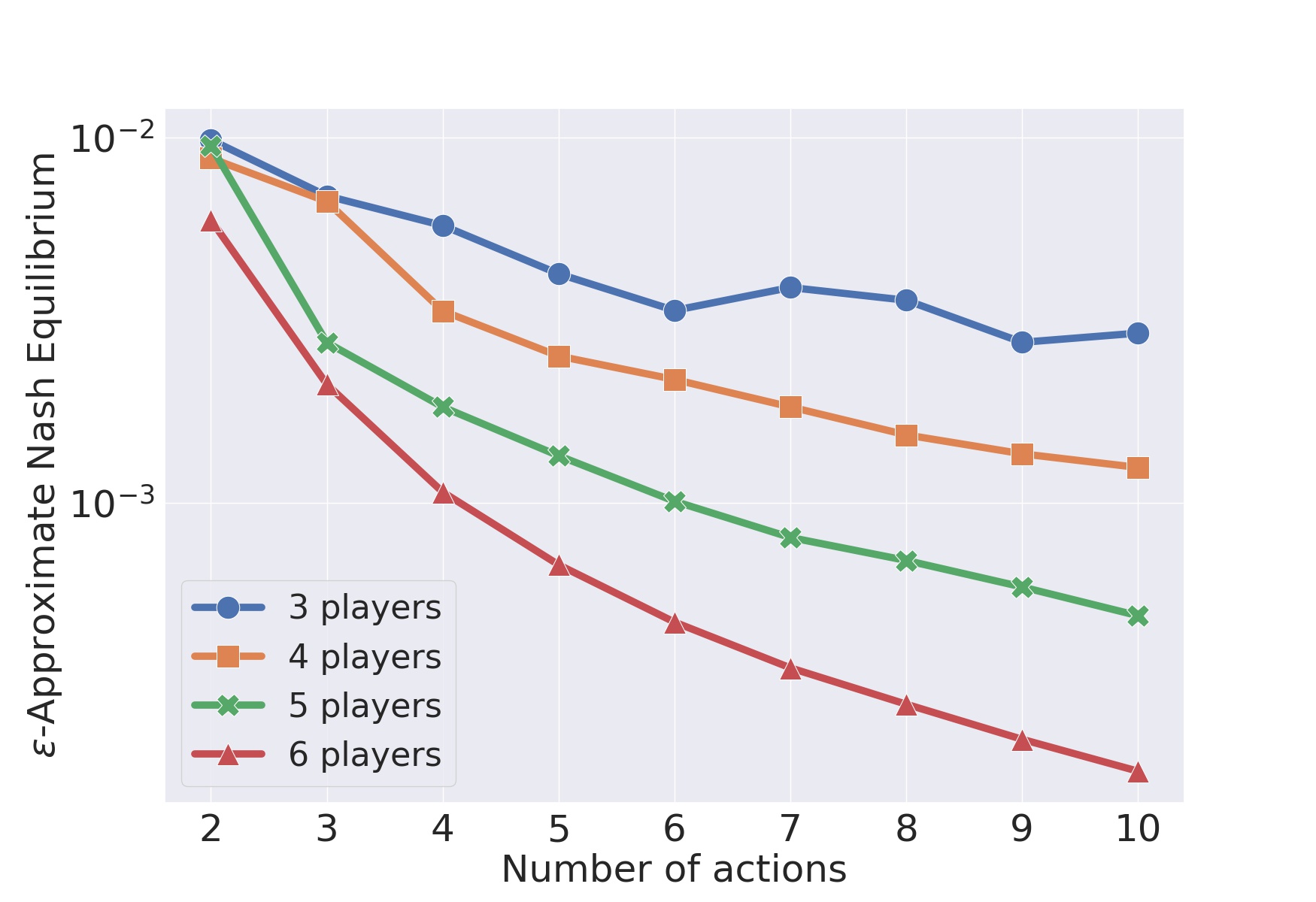}
\footnotesize{(b)}
\end{minipage}
\label{fig:robust}
\caption{(a) Comparison on random games with increasing game size. (b) NashD gradient descent on random games with increasing number of players and actions.
}
\label{fig:result}
\end{figure*}

\subsubsection{Random Game}
Random games are generated by sampling the utilities for every pure strategy profile from uniform distribution on $[0, 1]$.
We fix the number of actions to be $10$ and set the number of players ranging from $2$ to $6$. By doing so, we generate random games with game size increasing from $10^2$ to $10^6$. We also add a 2-player 3-action case (which has the game size of 9) to represent the game size of $10^1$. For each game size, we apply our method as well as baselines on $100$ game instances and take the average to compare the approximation performance as shown in Fig. \ref{fig:result} (a). 

Except for the 2-player 3-action case, our algorithm achieves the smallest approximation as game size increases.
TS provides similar results as ours on 2-player random games. We do not put it in the figure because TS is not applicable on the other N-player games.
As we can see, RM fails for most cases, which may due to the zero-sum invalidation and the large game sizes.
Further, we test the robustness of our approach by changing the number of players from $3$ to $6$ and increasing the number of actions from $2$ to $10$ for each player case. As shown in Fig. \ref{fig:result} (b), our algorithm maintains a small approximation regardless of the rise of the game.
It is interesting to notice that in Fig. \ref{fig:result} (a) and (b), the approximation might become even better as the game size grows up.
\cite{fearnley2015empirical} found the similar discovery and conjectured that randomly-generated games cannot fully represent all game classes.

\begin{figure*}[t]
\center
\begin{minipage}{0.99\textwidth}
\centering
\includegraphics[width=1.0\columnwidth]{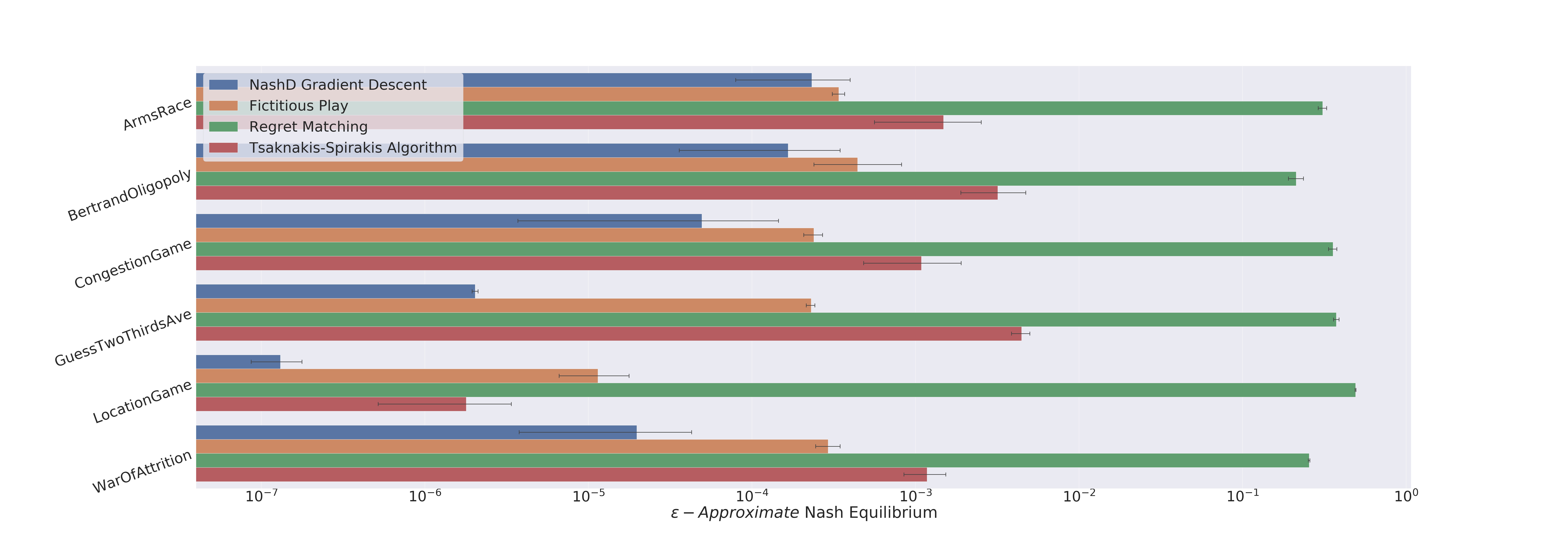}
\footnotesize{(a)}
\end{minipage}

\begin{minipage}{0.99\textwidth}
\centering
\includegraphics[width=1.0\columnwidth]{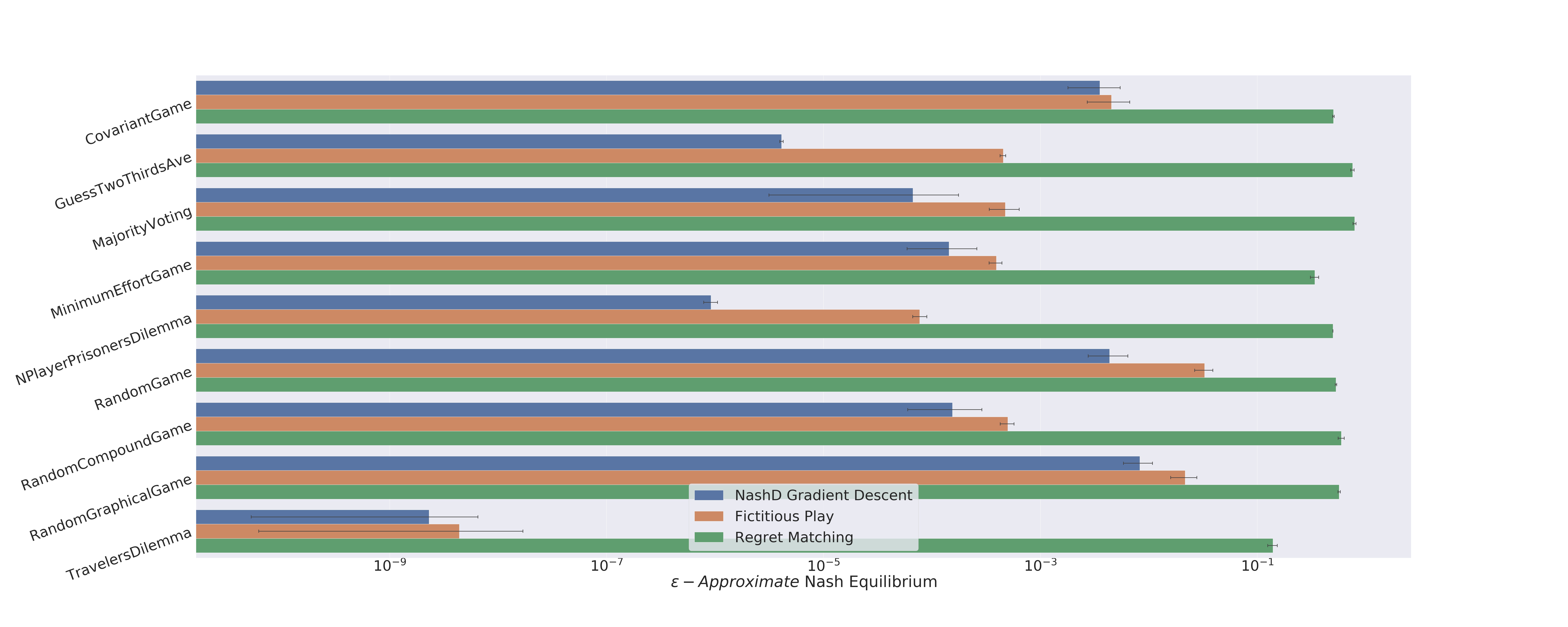}
\footnotesize{(b)}
\end{minipage}

\caption{Results on 2-player (a) and N-player (b) GAMUT game (the histogram illustrates the average approximation value over 100 game instances while the bars on it illustrate 0.95 confidence interval).}
\label{fig:result_gamut}
\end{figure*}

\subsubsection{GAMUT Game}

Finally, we test our algorithm on the comprehensive game suite GAMUT \cite{nudelman2004run} which contains various classes of games extracted from typical literature. 
Since TS is only applicable for 2-player games, we divide the game classes into two groups.
For the 2-player game group, we fix the number of players ($2$), and select the number of actions for each player randomly (range from $2$ to $5$).
For the N-player game group, we select both the number of players and the number of actions randomly (range from $3$ to $5$).
For each game class in GAMUT, we run our approach and baselines on 100 game instances and compare the average approximation value (shown in Fig. \ref{fig:result_gamut}).

As discussed by \cite{nudelman2004run}, the performance of game-theoretical algorithms always show large variation across different classes of GAMUT games.
Therefore, no single algorithm can behave as the best solution on all of the GAMUT games.
In spite of this intrinsic performance fluctuation across games, our approach obtains the best solutions on 9 types of N-player games, including the famous generalized N-player prisoners' dilemma game, majority voting game and congestion game.
Moreover, our algorithm always performs better on games that introduce randomness, as shown in random compound game and random graphical game.
Different from the random games we discussed above, compound game and graphical game have special properties which have been of importance in many applications.
It is in this sense we propose that gradient-based algorithms deserve more exploration on its fitted classes of games.
It is worthy to notice that FP performs fairly well in most game instances which may explain why it was used such widely in adversarial scenarios.
On the contrast, RM cannot be regarded as a proper Nash equilibrium solver, since most strategies it obtains are not close enough to any Nash equilibrium in general games.

\section{Conclusion}
In this paper, we introduce a new measure to characterize the distance between any strategy profile and a Nash equilibrium.
The value of the distance function is always non-negative and its zero-points are precisely coincident with the Nash equtbrium.
Based on that, we design a gradient descent algorithm to find approximate Nash equilibrium for any N-player normal-form games with discrete action space.
We prove the convergence of our method and show its practical performance across GAMUT game suite.
\bibliographystyle{plain}  
\bibliography{references}    

\begin{thebibliography}{10}

\bibitem{amar2015game}
Haitham~M Amar and Otman~A Basir.
\newblock A game theoretic approach for the ride-sourcing territory sharing problem.
\newblock In {\em 2015 IEEE 82nd Vehicular Technology Conference (VTC2015-Fall)}, pages 1--2, 2015.

\bibitem{berger2007brown}
Ulrich Berger.
\newblock Brown's original fictitious play.
\newblock {\em Journal of Economic Theory}, 135(1):572--578, 2007.

\bibitem{cesa2006prediction}
Nicolo Cesa-Bianchi and G{\'a}bor Lugosi.
\newblock {\em Prediction, learning, and games}.
\newblock Cambridge university press, 2006.

\bibitem{chen2009settling}
Xi~Chen, Xiaotie Deng, and Shang-Hua Teng.
\newblock Settling the complexity of computing two-player nash equilibria.
\newblock {\em Journal of the ACM (JACM)}, 56(3):14, 2009.

\bibitem{chong2013introduction}
Edwin~KP Chong and Stanislaw~H Zak.
\newblock {\em An introduction to optimization}, volume~76.
\newblock John Wiley \& Sons, 2013.

\bibitem{conitzer2009approximation}
Vincent Conitzer.
\newblock Approximation guarantees for fictitious play.
\newblock In {\em 2009 47th Annual Allerton Conference on Communication, Control, and Computing (Allerton)}, pages 636--643. IEEE, 2009.

\bibitem{deligkas2016empirical}
Argyrios Deligkas, John Fearnley, Tobenna~Peter Igwe, and Rahul Savani.
\newblock An empirical study on computing equilibria in polymatrix games.
\newblock In {\em Proceedings of the 2016 International Conference on Autonomous Agents \& Multiagent Systems}, pages 186--195. International Foundation for Autonomous Agents and Multiagent Systems, 2016.

\bibitem{deligkas2017computing}
Argyrios Deligkas, John Fearnley, Rahul Savani, and Paul Spirakis.
\newblock Computing approximate nash equilibria in polymatrix games.
\newblock {\em Algorithmica}, 77(2):487--514, 2017.

\bibitem{fearnley2015empirical}
John Fearnley, Tobenna~Peter Igwe, and Rahul Savani.
\newblock An empirical study of finding approximate equilibria in bimatrix games.
\newblock In {\em International Symposium on Experimental Algorithms}, pages 339--351. Springer, 2015.

\bibitem{gilpin2007gradient}
Andrew Gilpin, Samid Hoda, Javier Pena, and Tuomas Sandholm.
\newblock Gradient-based algorithms for finding nash equilibria in extensive form games.
\newblock In {\em International Workshop on Web and Internet Economics}, pages 57--69. Springer, 2007.

\bibitem{goodfellow2014generative}
Ian Goodfellow, Jean Pouget-Abadie, Mehdi Mirza, Bing Xu, David Warde-Farley, Sherjil Ozair, Aaron Courville, and Yoshua Bengio.
\newblock Generative adversarial nets.
\newblock {\em Advances in neural information processing systems}, 27:2672--2680, 2014.

\bibitem{hart2000simple}
Sergiu Hart and Andreu Mas-Colell.
\newblock A simple adaptive procedure leading to correlated equilibrium.
\newblock {\em Econometrica}, 68(5):1127--1150, 2000.

\bibitem{heinrich2015fictitious}
Johannes Heinrich, Marc Lanctot, and David Silver.
\newblock Fictitious self-play in extensive-form games.
\newblock In {\em International Conference on Machine Learning}, pages 805--813, 2015.

\bibitem{kang2017sequential}
Qiyu Kang and Wee~Peng Tay.
\newblock Sequential multi-class labeling in crowdsourcing: a ulam-renyi game approach.
\newblock In {\em Proceedings of the International Conference on Web Intelligence}, pages 245--251. ACM, 2017.

\bibitem{krishna1998convergence}
Vijay Krishna and Tomas Sj{\"o}str{\"o}m.
\newblock On the convergence of fictitious play.
\newblock {\em Mathematics of Operations Research}, 23(2):479--511, 1998.

\bibitem{lemke1964equilibrium}
Carlton~E Lemke and Joseph~T Howson, Jr.
\newblock Equilibrium points of bimatrix games.
\newblock {\em Journal of the Society for industrial and Applied Mathematics}, 12(2):413--423, 1964.

\bibitem{lockhart2019computing}
Edward Lockhart, Marc Lanctot, Julien P{\'e}rolat, Jean-Baptiste Lespiau, Dustin Morrill, Finbarr Timbers, and Karl Tuyls.
\newblock Computing approximate equilibria in sequential adversarial games by exploitability descent.
\newblock {\em arXiv preprint arXiv:1903.05614}, 2019.

\bibitem{maehara2015budget}
Takanori Maehara, Akihiro Yabe, and Ken-ichi Kawarabayashi.
\newblock Budget allocation problem with multiple advertisers: A game theoretic view.
\newblock In {\em ICML}, volume~32, pages 428--437, 2015.

\bibitem{mazumdar2019policy}
Eric Mazumdar, Lillian~J Ratliff, Michael~I Jordan, and S~Shankar Sastry.
\newblock Policy-gradient algorithms have no guarantees of convergence in continuous action and state multi-agent settings.
\newblock {\em arXiv preprint arXiv:1907.03712}, 2019.

\bibitem{nash1950equilibrium}
John~F Nash.
\newblock Equilibrium points in n-person games.
\newblock {\em Proceedings of the national academy of sciences}, 36(1):48--49, 1950.

\bibitem{niyato2016economics}
Dusit Niyato, Xiao Lu, Ping Wang, Dong~In Kim, and Zhu Han.
\newblock Economics of internet of things: An information market approach.
\newblock {\em IEEE Wireless Communications}, 23(4):136--145, 2016.

\bibitem{nudelman2004run}
Eugene Nudelman, Jennifer Wortman, Yoav Shoham, and Kevin Leyton-Brown.
\newblock Run the gamut: A comprehensive approach to evaluating game-theoretic algorithms.
\newblock In {\em Proceedings of the Third International Joint Conference on Autonomous Agents and Multiagent Systems-Volume 2}, pages 880--887. IEEE Computer Society, 2004.

\bibitem{Papadimitriou2018from}
Christos Papadimitriou and Georgios Piliouras.
\newblock From nash equilibria to chain recurrent sets: An algorithmic solution concept for game theory.
\newblock {\em Entropy}, 20(10), 2018.

\bibitem{porter2008simple}
Ryan Porter, Eugene Nudelman, and Yoav Shoham.
\newblock Simple search methods for finding a nash equilibrium.
\newblock {\em Games and Economic Behavior}, 63(2):642--662, 2008.

\bibitem{sandholm2005mixed}
Tuomas Sandholm, Andrew Gilpin, and Vincent Conitzer.
\newblock Mixed-integer programming methods for finding nash equilibria.
\newblock In {\em AAAI}, pages 495--501, 2005.

\bibitem{tsaknakis2007optimization}
Haralampos Tsaknakis and Paul~G Spirakis.
\newblock An optimization approach for approximate nash equilibria.
\newblock In {\em International Workshop on Web and Internet Economics}, pages 42--56. Springer, 2007.

\bibitem{tsaknakis2008performance}
Haralampos Tsaknakis, Paul~G Spirakis, and Dimitrios Kanoulas.
\newblock Performance evaluation of a descent algorithm for bi-matrix games.
\newblock In {\em International Workshop on Internet and Network Economics}, pages 222--230. Springer, 2008.

\bibitem{zinkevich2008regret}
Martin Zinkevich, Michael Johanson, Michael Bowling, and Carmelo Piccione.
\newblock Regret minimization in games with incomplete information.
\newblock In {\em Advances in neural information processing systems}, pages 1729--1736, 2008.

\end{thebibliography}

\end{document}